\documentclass[preprint,5p,times,twocolumn,authoryear]{elsarticle}

\usepackage{lineno,hyperref}

\usepackage[cmex10]{amsmath}
\usepackage{amssymb}
\usepackage{mathrsfs}
\usepackage{amsthm}

\usepackage{graphicx}
\usepackage{color}
\usepackage{algorithm}
\usepackage{algorithmic}
\usepackage{bm}
\usepackage{setspace}

\DeclareGraphicsExtensions{.pdf,.png,.jpg,.eps,.ps}

\def\ba{\begin{array}}
	\def\ea{\end{array}}
\def\baa{\begin{align}}
	\def\eaa{\end{align}}

\newcommand{\bsq}{\begin{subequations}}
	\newcommand{\esq}{\end{subequations}}

\newcommand{\beq}{\begin{equation}}
\newcommand{\eeq}{\end{equation}}
\newcommand{\bq}{\begin{eqnarray}}
\newcommand{\eq}{\end{eqnarray}}
\newcommand{\bqn}{\begin{eqnarray*}}
	\newcommand{\eqn}{\end{eqnarray*}}
\newcommand{\bee}{\begin{enumerate}}
	\newcommand{\eee}{\end{enumerate}}
\newcommand{\bi}{\begin{itemize}}
	\newcommand{\ei}{\end{itemize}}

\usepackage{comment}
\usepackage{ifthen}
\newboolean{showcomments}
\setboolean{showcomments}{true}
\newcommand{\wang}[1]{\ifthenelse{\boolean{showcomments}}
	{ \textcolor[rgb]{1,0,1}{(ZW:  #1)}}{}}
\newcommand{\fliu}[1]{\ifthenelse{\boolean{showcomments}}
	{ \textcolor{red}{(FL:  #1)}}{}}
\newcommand{\slow}[1]{\ifthenelse{\boolean{showcomments}}
	{ \textcolor{blue}{(SL:  #1)}}{}}

\theoremstyle{definition}
\newtheorem{theorem}{Theorem}
\newtheorem{lemma}[theorem]{Lemma}

\theoremstyle{definition}
\newtheorem{definition}{Definition}
\newtheorem{remark}{Remark}

{\newtheorem{assumption}{\textit{Assumption}}}

\allowdisplaybreaks

\journal{xxx}

\begin{document}
\setstretch{0.996}
\begin{frontmatter}

\title{Asynchronous Distributed Voltage Control in Active Distribution Networks}

\tnotetext[mytitlenote]{This work was supported  by the National Natural Science Foundation of China ( No. 51677100, U1766206, No. 51621065).}

\author[thu]{Zhaojian~Wang}
\author[thu]{Feng~Liu\corref{mycorrespondingauthor}}\ead{lfeng@mail.tsinghua.edu.cn}
\author[thu]{Yifan Su}
\author[xjtu]{Boyu Qin}
%\author[thu]{Shengwei Mei}

\cortext[mycorrespondingauthor]{Corresponding author}

\address[thu]{State Key Laboratory of Power Systems, Department of Electrical Engineering, Tsinghua University, Beijing 100084, China}
\address[xjtu]{State Key Laboratory of Electrical Insulation and Power Equipment, School of Electrical Engineering, Xi'an Jiaotong University, Xi'an 710049, China}

\begin{abstract}
	With  the explosion of distributed energy resources (DERs), voltage regulation in distribution networks has been facing  a great challenge. This paper derives an asynchronous distributed voltage control strategy based on the partial primal-dual gradient algorithm, where both active and reactive controllable power of DERs are considered.     
	Different types of asynchrony due to imperfect communication or practical limits, such as random time delays and non-identical sampling/control rates, are  fitted into a unified analytic framework. The asynchronous algorithm is then converted into a fixed-point problem by employing the operator splitting method, which leads to a convergence proof with mild conditions. Moreover, an online implementation method is provided to make the controller adjustable to time-varying environments. Finally, numerical experiments  are carried out on a rudimentary 8-bus system and the IEEE-123 distribution network to verify the effectiveness of the proposed method.
\end{abstract}

\end{frontmatter}

\section{Introduction}

With the proliferation of distributed energy resources (DERs), such as small hydro plants, Photovoltaics (PVs) and energy storage systems, voltage regulation in active distribution networks is greatly challenged, On the one hand, the voltage quality  remarkably degrades, e.g., the voltage may fluctuate rapidly due to the variation of renewable generations and over-voltage exists at the buses DERs connected. On the other hand,  many DERs, such as some small hydro plants \cite {han2014optimization} and inverter-integrated DERs \cite {turitsyn2011options}, have great potential of  voltage regulation by appropriately managing their active or reactive power outputs. Beyond the capability of traditional voltage regulation schemes, these challenges call for a new voltage control paradigm. 

The voltage control in a distribution network aims to minimize the voltage mismatch by regulating active or reactive power outputs of controllable DERs. Generally speaking, it can be viewed as a type of optimal power flow (OPF) problems, where the branch power flow model is usually utilized \cite{baran1989optimal2,baran1989optimal}. Similar topics have been studied extensively in the literature. Related works can roughly be categorized into three classes in terms of the communication requirements: centralized control,  local control and distributed control. In the centralized voltage control, a global optimization problem is formulated and solved by a central controller to determine optimal set-points for the overall system \cite{farivar2011inverter,farivar2012optimal,kekatos2015stochastic}. In this case, the central controller collects all the required information and communicates with all DERs. However, it suffers from the single-point-failure issue and costs long computation time when the number of DERs is  large. 
As for the local voltage control, locally available information such as bus voltage magnitude is utilized to design the controller \cite{turitsyn2011options}. In the problem formulation, the linearized distribution power flow is usually utilized, and the objective function is a specific form \cite{zhu2016fast, liu2017decentralized, zhou2018reverse}. 
%A more general framework is proposed in \cite{zhu2016fast} for developing local voltage control strategies that explicitly account for the reactive power capacity limits. This method has been improved  to adjust to a  changing environment in \cite{liu2017decentralized}, and reformulated from the perspective of reverse and forward engineering in \cite{zhou2018reverse}.
As it uses only local information, the response is  rapid. However, the control objective is restricted to specific types, making it less flexible.
The distributed voltage control can avoid the disadvantages of centralized and local controls to some extent \cite{antoniadou2017distributed}. Compared with the centralized control, there is no central controller and communication is usually between immediate neighbors \cite{vsulc2014optimal, bolognani2015distributed, zhang2015optimal, liu2018distributed, liu2018hybrid} or two-hop neighbors \cite{Tang2018Fast}. Compared with the local voltage control, the objective function can be more general and practical. In existing literature, the distributed voltage control is usually synchronous. However, asynchrony widely exists in power systems, such as communication time delay caused by congestion or even failure and different sampling or computation rates. In the synchronous case, the slowest bus and communication channel may cripple the system \cite{peng2016arock, yi2018asynchronous}.

This paper designs an asynchronous distributed strategy for voltage control in distribution networks. Various types of asynchrony in power systems are considered, such as communication delay, and different sampling rates, which are fitted into a unified framework. This is different from \cite{bolognani2015distributed}, which only considers asynchronous iterations and assumes no communication delay. The proposed method is also different from the asynchronous control in \cite{zhu2016fast}, which is local but with restriction on the objective function. 
In this paper, we consider the regulation of both active and reactive controllable power of DERs. In the controller design, partial primal-dual gradient algorithm is utilized, which is formulated as the form of the Krasnosel'ski{\v{i}}-Mann iteration. In this way, the objective function is only required to be \emph{convex} and have a Lipschitzian gradient, which relaxes the assumption commonly used in most of existing literature (\emph{strong convexity} is required). Moreover, the operator splitting method is employed to convert the control algorithm into a fixed-point problem, which greatly simplifies  the convergence proof. In terms of practical application, a online implementation method is provided to make the controller adjustable to time-varying environments.
%{\color{red}the advantage of difference} 

%\subsection{Paper Organization}
The rest of this paper is organized as follows. In Section 2, we introduce some preliminaries and the model of distribution networks. Section 3 formulates the optimal voltage control problem. The asynchronous controller is investigated in Section 4. In Section 5, convergence and optimality of the equilibrium are proved. Section 6 introduces the implementation of the proposed method. We confirm the performance of  controllers via simulations on an 8-bus system and IEEE 123-bus system in Section 7. Section 8 concludes the paper.

%\textcolor{red}{.}
\section{Preliminaries and System Modeling}

\subsection{Preliminaries} 

In this paper, $\mathbb{R}^n$ ($\mathbb{R}^n_{+}$) is the $n$-dimensional (nonnegative) Euclidean space. For a column vector $x\in \mathbb{R}^n$ (matrix $A\in \mathbb{R}^{m\times n}$), $x^{{T}}$($A^{{T}}$) denotes its transpose.  For vectors $x,y\in \mathbb{R}^n$, $x^{{T}}y=\left\langle x,y \right\rangle$ denotes the inner product of $x,y$. 
$\left\|x \right\|=\sqrt{x^{{T}}x}$ denotes the norm induced by the inner product. 
For a positive definite matrix $G$, denote the inner product $\left\langle x,y \right\rangle_G=\left\langle Gx,y \right\rangle$. Similarly, the $G$-matrix induced norm $\left\|x \right\|_G=\sqrt{\left\langle Gx,x \right\rangle}$. 
Use $I_n$ to denote the identity matrix with dimension $n$. Sometimes, we also omit $n$ and use $I$ to denote the identity matrix with proper dimension if there is no confusion. For a matrix $A=[a_{ij}]$, $a_{ij}$ stands for the entry in the $i$-th row and $j$-th column of $A$. Use $\prod_{i=1}^n\Omega_i$ to denote the Cartesian product of the sets $\Omega_i, i=1, \cdots, n$. 
Given a collection of $y_i$ for $i$ in a certain set $Y$, define $\text{col}(y_{j}):=(y_1, y_2, \cdots, y_n)^T$ and denote its vector form by $\textbf{y}:=\text{col}(y_j)$. Define the projection of $x$ onto a set $\Omega$ as 
\begin{equation}
\label{def:projection}
\setlength{\abovedisplayskip}{4pt}	
\setlength{\belowdisplayskip}{4pt}
\mathcal{P}_{\Omega}(x)=\arg \min_{y\in \Omega}\left\|x-y \right\|
\end{equation}
Use ${\rm{Id}}$ to denote the identity operator, i.e., ${\rm{Id}}(x)=x$, $\forall x$. Define the normal cone as $N_\Omega(x)=\{v|\left\langle v, y-x\right\rangle\le 0, \forall y\in \Omega\}$. We have $\mathcal{P}_{\Omega}(x)=({\rm{Id}}+N_{\Omega})^{-1}(x)$ \cite{yi2017distributed}, \cite[Chapter 23.1]{bauschke2011convex}.

%
%A set $\Omega$ is a convex set is $\lambda x+ (1-\lambda) y \in \Omega,\ \forall \lambda\in [0,\ 1],\ \forall x,y \in \Omega$. An extended value proper function $f:\mathbb{R}^m\rightarrow\mathbb{R}$ is a convex function if 
%$f(\lambda x+ (1-\lambda) y)\le \lambda f(x)+(1-\lambda)f(y),\ \forall \lambda\in [0,\ 1],\ \forall x,y \in dom (f)$ 
%with $dom (f)$ as the domain of $f$. 

%A mapping $F:\mathbb{R}^n\rightarrow \mathbb{R}^n$ is called monotone on a set $\Omega$ if $(x_1-x_0)^{\mathrm{T}}(F(x_1)-F(x_0))\ge0$ for all $x_0,x_1\in \Omega$. $F$ is strictly monotone if the inequality is strict when $x_0\neq x_1$. 

For a set-valued operator $\mathcal{U}:\mathbb{R}^n\rightarrow2^{\mathbb{R}^n}$, its domain is $\text{dom}\mathcal{U}:=\{x\in \mathbb{R}^n|\mathcal{U}x\neq \emptyset\}$. The graph of $\mathcal{U}$ is defined as $\text{gra}\mathcal{U}:=\{(x,u)\in \mathbb{R}^n\times\mathbb{R}^n|u\in\mathcal{U}x \}$. An operator $\mathcal{U}$ is monotone if $\forall(x,u), \forall(y,v)\in \text{gra}\mathcal{U}$, we have $\left\langle x-y, u-v\right\rangle\ge 0$. It is called maximally monotone if $\text{gra}\mathcal{U}$ is not strictly contained in the graph of any other monotone operator.
For a single-valued operator $\mathcal{T}:\Omega\subset\mathbb{R}^n\rightarrow\mathbb{R}^n$, a point $x\in \Omega$ is a fixed point of $\mathcal{T}$ if $\mathcal{T}(x)\equiv x$. The set of fixed points of $\mathcal{T}$ is denoted by ${Fix}(\mathcal{T})$. $\mathcal{T}$ is nonexpansive if $\left\|\mathcal{T}(x)- \mathcal{T}(y)\right\|\le\left\|x- y\right\|, \forall x, y \in \Omega$. $\mathcal{T}$ is firmly nonexpansive if $\left\|\mathcal{T}(x)- \mathcal{T}(y)\right\|^2+\left\|({\rm{Id}}-\mathcal{T})(x)- ({\rm{Id}}-\mathcal{T})(y)\right\|^2\le\left\|x- y\right\|^2, \forall x, y \in \Omega$. For $\alpha\in (0,1)$, $\mathcal{T}$ is called $\alpha$-averaged if there exists a nonexpansive operator $\mathcal{T}$ such that $\mathcal{T}=(1-\alpha){\rm{Id}}+\alpha \mathcal{T}$. We use $\mathcal{A}(\alpha)$ to denote the $\alpha$-averaged operators. For $\beta\in\mathbb{R}^1_{+}$, $\mathcal{T}$ is called $\beta$-cocoercive if $\beta\mathcal{T}\in \mathcal{A}(\frac{1}{2})$.

\subsection{System Modeling}

Consider a radial distribution network with $(n+1)$ buses collected in the set $\mathcal{N}_0:=\{0\}\cup\mathcal{N}$, where $\mathcal{N}:=\{1,\cdots, n\}$ and bus 0 is the substation bus (slack bus) and is assumed to have a fixed voltage $U_0$. Lines are denoted by the set $\mathcal{E}:=\{(i, j)\}\subset\mathcal{N}\times\mathcal{N}$. 
Due to the tree topology, the cardinality of $|\mathcal{E}| = n$. Use $N_j$ and $N_j^2$ to denote the neighbors and two-hop neighbors of bus $j$ respectively.  

For bus $j$, use $U_j$ to denote its voltage magnitude. Use $p_j$ and $q_j$ to denote the active and reactive power generations respectively, which are controllable. $p_j^c$ and $q_j^c$ are active and reactive power loads, which are uncontrollable. For line $(i, j)\in \mathcal{E}$, use $r_{ij}$ and $x_{ij}$ to denote its line resistance and reactance. The active and reactive power from bus $i$ to $j$ is denoted by $P_{ij}$ and $Q_{ij}$ respectively. The linearized DistFlow equations  are given as \cite{baran1989optimal2,baran1989optimal,zhu2016fast}
%\begin{subequations}
%	\setlength{\abovedisplayskip}{4pt}	
%	\setlength{\belowdisplayskip}{4pt}
%	\label{DistFlow}
%	\begin{align}
%	\label{Active_Power}
%	P_{ij}+p_j-p_j^c&=\sum\nolimits_{k \in  N_j }P_{jk}+r_{ij}\ell_{ij}\\
%	\label{Reactive_Power}
%	Q_{ij}+q_j-q_j^c&=\sum\nolimits_{k \in  N_j }Q_{jk}+x_{ij}\ell_{ij} \\
%	\label{V_Square}
%	U_i^2-U_j^2 &=2(r_{ij}P_{ij}+x_{ij}Q_{ij})-\left(r_{ij}^2+x_{ij}^2\right)\ell_{ij}
%	\\
%	\label{S_Square}
%	\ell_{ij}U_i^2&=P_{ij}^2+Q_{ij}^2
%	\end{align}
%\end{subequations}
%where $\ell_{ij}$ is used to denote the current square from bus $i$ to $j$. Assume the loss is negligible compared to line flow, and define $V_i:=\frac{U_i^2}{2}$. We have the following linearized DistFlow  model \cite{zhu2016fast}.
\begin{subequations}
	\label{LinDistFlow}
	\setlength{\abovedisplayskip}{4pt}	
	\setlength{\belowdisplayskip}{4pt}
	\begin{align}
	\label{LinActive_Power}
	P_{ij}+p_j-p_j^c&=\sum\nolimits_{k \in  N_j }P_{jk}\\
	\label{LinReactive_Power}
	Q_{ij}+q_j-q_j^c&=\sum\nolimits_{k \in  N_j }Q_{jk} \\
	\label{LinV}
	U_i^2-U_j^2 &=r_{ij}P_{ij}+x_{ij}Q_{ij}
	\end{align}
\end{subequations}

The relative error of the linearization is very small, at the order of 1\% \cite{zhu2016fast}. 
Denote the incidence matrix of the network $(\mathcal{N}_0, \mathcal{E})$ by $\mathcal{M}\in\mathbb{R}^{(n+1)\times n}$. Moreover, use $\textbf{m}_0^T$ to denote the first row of $\mathcal{M}$, while the rest of the matrix is denoted by $\textbf{M}$. Define $V_i:=\frac{U_i^2}{2}$, and then the compact form of \eqref{LinDistFlow} is 
\begin{subequations}\label{compact_LinDistFlow}
	\setlength{\abovedisplayskip}{4pt}	
	\setlength{\belowdisplayskip}{4pt}
	\begin{align}
	\label{compact_LinActive_Power}
	\textbf{M}\textbf{P}&=\textbf{p}-\textbf{p}^c\\
	\label{compact_LinReactive_Power}
	\textbf{M}\textbf{Q}&=\textbf{q}-\textbf{q}^c \\
	\label{compact_LinV}
	[\textbf{m}_0,\ \textbf{M}^T]\cdot[V_0,\ \textbf{V}^T]^T &=\text{diag}(\textbf{r})\textbf{P}+\text{diag}(\textbf{x})\textbf{Q}
	\end{align}
\end{subequations}
where $\text{diag}(\textbf{r})$ is the diagonal matrix composed of $r_{ij}$, and similar is $\text{diag}(\textbf{x})$. 
As the network is connected, the rank of $\mathcal{M}$ is $n$. Thus, $\textbf{M}$ is of full rank and invertible \cite{zhu2016fast}. Finally, we have
\begin{equation}
\setlength{\abovedisplayskip}{1.5pt}	
\setlength{\belowdisplayskip}{1.5pt}
\textbf{V}=\textbf{R}\textbf{p}+\textbf{X}\textbf{q}-\textbf{M}^{-T}\textbf{m}_0V_0-\textbf{R}\textbf{p}^c-\textbf{X}\textbf{q}^c
\end{equation}
where $\textbf{R}=\textbf{M}^{-T}\text{diag}(\textbf{r})\textbf{M}^{-1}$ and $\textbf{X}=\textbf{M}^{-T}\text{diag}(\textbf{x})\textbf{M}^{-1}$. $\textbf{R}$ and $\textbf{X}$ are all symmetric positive definite matrices. By \cite{kekatos2015fast,kekatos2016voltage}, we know $-\textbf{M}^{-T}\textbf{m}_0=\textbf{1}_n$. Denote the inverse of $\textbf{X}$ by $\textbf{B}=\textbf{M}\text{diag}(\textbf{x}^{-1})\textbf{M}^{T}$, which is also positive definite. It is proved in \cite[Theorem 2]{zhou2018reverse} that $\textbf{B}=\textbf{L}+\text{diag}(\frac{1}{x_{0j}})$, where $\textbf{L}$ is the weighted Laplacian matrix of the subtree (i.e., without bus 0) and $x_{0j}$ is the reactance of the line connected to bus 0. If bus $j$ is not connected to the bus $0$ directly, $x_{0j}=\infty$.

If the distribution network lines have unified resistance-reactance ratio, i.e. there exists a constant $K=\frac{r_{ij}}{x_{ij}}, \forall (i, j)\in \mathcal{E}$, the network is called homogeneous. For a homogeneous network, we have $\textbf{R}=K\textbf{X}$ and $\textbf{B}\cdot\textbf{R}=K$. In the  analysis of this paper, it is assumed that the distribution network is homogeneous, which is true in the most cases in practice \cite{bolognani2015distributed, Tang2018Fast}. In the simulation, however, we also use the heterogeneous network to verify the performance of the controller. 
%It should be noted that the system in the simulation is not homogeneous and the AC power flow is used. 

\section{Problem Formulation}

The optimal voltage control for a homogeneous network is formulated as
\begin{subequations}\label{eq_opt2}  
	\setlength{\abovedisplayskip}{4pt}	
	\setlength{\belowdisplayskip}{4pt}
	\begin{align}
	\min\limits_{\textbf{V}, \textbf{p}, \textbf{q}} &\quad   f=\frac{1}{2}\|\textbf{V}-\textbf{V}^o\|^2 + \sum\nolimits_{j\in\mathcal{N}}g_j(p_j,q_j) 
	\label{eq_opt2a}
	\\ 
	\text{s.t.}  
	&  \quad	\textbf{B}\textbf{V}=K\textbf{p}+\textbf{q}+\bm{\varpi}^s
	\label{eq_opt2b} \\
	&  \quad	\underline{{p}}_j\le {p}_j\le \overline{{p}}_j,\ \forall j\label{eq_opt2d} \\
	&  \quad	\underline{{q}}_j\le {q}_j\le \overline{{q}}_j,\ \forall j\label{eq_opt2e} \\
	&  \quad	0\le {p}_j^2+{q}_j^2\le s_j^2,\ \forall j\label{eq_opt2f} 
	\end{align}   
\end{subequations}
where $\bm{\varpi}^s=\textbf{B}(-\textbf{M}^{-T}\textbf{m}_0V_0-\textbf{R}\textbf{p}^c-\textbf{X}\textbf{q}^c)$. $\textbf{V}^o$ is the desired voltage profile and set as $\textbf{V}^o=0.5\times\textbf{1}_n$. $\underline{{p}}_j, \overline{{p}}_j$ are lower and upper bounds of ${p}_j$. $\underline{{q}}_j, \overline{{q}}_j$ are lower and upper bounds of ${q}_j$. $s_j$ is the apparent power capability of the inverter. 
The objective function is composed of two parts: voltage difference square, active and reactive power cost, for which we make an assumption.

\begin{assumption}\label{Lipschitz}
	The function $g_i(x)$ is convex, and $\nabla g_i(x)$ is $\vartheta$-Lipschitzian, i.e., for some $\vartheta>0$, $\|\nabla g_i(x_1)-\nabla g_i(x_2)\|\le \vartheta \|x_1-x_2\|, \forall x_1, x_2$.
\end{assumption}

For each bus $j$, the feasible region is defined as 
\begin{equation*}
\setlength{\abovedisplayskip}{4pt}	
\setlength{\belowdisplayskip}{4pt}
\Omega_j=\{\ (p_j, q_j)\ |\ p_j, q_j\ \text{satisfy}\ \eqref{eq_opt2d}, \eqref{eq_opt2e},\eqref{eq_opt2f}\ \}
\end{equation*}
%$\Omega_j$ is not a box, which is different from existing literatures. Here, we give a schematic diagram of $\Omega_j$ in Fig.\ref{fig_feasible_region}.
%\begin{figure}[t]
%	\centering
%	\includegraphics[width=0.2\textwidth]{feasible_region}
%	\caption{The feasible region of each bus}
%	\label{fig_feasible_region}
%\end{figure}

The Lagrangian of \eqref{eq_opt2} is 
\begin{equation}\label{Lagrangian}
\setlength{\abovedisplayskip}{4pt}	
\setlength{\belowdisplayskip}{4pt}
\begin{split}
\mathop{\mathcal{L}(\textbf{V}, \textbf{p}, \textbf{q}, \bm{\lambda})}\limits_{(\textbf{p}, \textbf{q})\in{\Omega}} =\frac{1}{2}\|\textbf{V}-\textbf{V}^*\|^2 + \sum\nolimits_{j\in\mathcal{N}}g_j(p_j,q_j) \\
\quad + \bm{\lambda}^T(\textbf{B}\textbf{V}-K\textbf{p}-\textbf{q}-\bm{\varpi}^s)
\end{split}
\end{equation}
where $\Omega=\prod_{j\in \mathcal{N}} \Omega_j$.

\begin{remark}[Objective function]
	The item $g_j(p_j, q_j)$ in the objective function is more general compared with existing literature, which is only required to be convex and has a Lipschitzian gradient instead of strongly convex.
	If the objective function is formulated by a $\textbf{B}$-induced norm, i.e., $f=\frac{1}{2}\|\textbf{V}-\textbf{V}^o\|^2_\textbf{B}$, we can design a local controller, as  done in \cite{zhu2016fast, liu2017decentralized, bolognani2015distributed}.
\end{remark}

In the problem \eqref{eq_opt2}, we consider the regulation of both active and reactive power of DERs. The main motivations are two-folds: first, some DERs such as many small hydro plants have no capability of reactive power regulation; second, the distribution networks have comparable resistance and reactance. Thus, regulating both active and reactive power turns to be necessary in the voltage control of active distribution networks.

\section{Asynchronous Voltage Control}

%\subsection{Asynchronous distributed algorithm}
In the asynchronous controller design, we adopt the partial primal-dual algorithm. Similar methods have been explored in \cite{Li:Connecting,wang2019distributed2} for a continuous-time setting and \cite{liu2018hybrid} for the discrete setting. However, \cite{liu2018hybrid} does not design an asynchronous algorithm. In a partial primal-dual algorithm, $\textbf{V}$ is obtained by solving the following problem.
\begin{equation}\label{partial}
\setlength{\abovedisplayskip}{4pt}	
\setlength{\belowdisplayskip}{4pt}
\textbf{V}_t=\arg\min_\textbf{V}\mathop{\mathcal{L}(\textbf{V}, \textbf{p}_t, \textbf{q}_t, \bm{\lambda}_t)}=-\textbf{B}^T{\bm{\lambda}}_{t}+{\textbf{V}}^o
\end{equation}
Define $\bm{\varpi}^a=\bm{\varpi}^s-\textbf{B}{\textbf{V}}^o$. 
Each bus has its own iteration number $t_j$, implying that a \emph{local clock} is used. Then, various types of asynchrony can be considered as time intervals between two iterations.
At $t_j$, bus $j$ computes in the following way, which has the form of the Krasnosel'ski{\v{i}}-Mann iteration.
\begin{subequations}\label{inertia_algorithm_async} 
	\setlength{\abovedisplayskip}{4pt}	
	\setlength{\belowdisplayskip}{4pt}
	\begin{align}
	&\left[ {\begin{array}{*{20}{c}}
		\tilde p_{j,t_j}\\
		\tilde q_{j,t_j}
		\end{array}} \right]=\mathcal{P}_{\Omega_j}\left[ {\left( {\begin{array}{*{20}{c}}
			p_{j,t_j-\tau_{j}^{t_j}}\\
			q_{j,t_j-\tau_{j}^{t_j}}
			\end{array}} \right)}\right.\nonumber\\
	&\qquad\quad\quad\left.-\alpha_{pq} 
	\left( {\begin{array}{*{20}{c}}
		\frac{\partial g_j}{\partial p_j }( p_{j,t_j-\tau_{j}^{t_j}}, q_{j,t_j-\tau_{j}^{t_j}})-K \lambda_{j,t_j-\tau_{j}^{t_j}}\\
		\frac{\partial g_j}{\partial q_j }(p_{j,t_j-\tau_{j}^{t_j}}, q_{j,t_j-\tau_{j}^{t_j}})- \lambda_{j,t_j-\tau_{j}^{t_j}}
		\end{array}} \right) \right] \label{ASYN1}\\
	&\tilde\lambda_{j,t_j}=\lambda_{j,t_j-\tau_{j}^{t_j}} + \alpha_\lambda\left(-\sum\nolimits_{k\in N_j\cup N_j^2}\tilde B_{jk} \lambda_{k,t_k-\tau_{k}^{t_k}}\right.\nonumber\\
	&\qquad\quad \left.-2K\tilde p_{j,t_j}-2\tilde q_{j,t_j} +K p_{j,t_j-\tau_{j}^{t_j}}+ q_{j,t_j-\tau_{j}^{t_j}}-\varpi_j^a\right)\label{ASYN2}\\
	&{\lambda}_{j,t_j+1}={\lambda}_{j,t_j-\tau_{j}^{t_j}}+\eta(\tilde{\lambda}_{j,t_j}-{\lambda}_{j,t_j-\tau_{j}^{t_j}}) \label{ASYN3}\\
	&{p}_{j,t_j+1}={p}_{j,t_j-\tau_{j}^{t_j}}+\eta(\tilde{p}_{j,t_j}-{p}_{j,t_j-\tau_{j}^{t_j}})\label{ASYN4} \\		
	&{q}_{j,t_j+1}={q}_{j,t_j-\tau_{j}^{t_j}}+\eta(\tilde{q}_{j,t_j}-{q}_{j,t_j-\tau_{j}^{t_j}}) \label{ASYN5}\\		
	&{{V}}_{j,t_j+1}= -\sum\nolimits_{k\in\mathcal{N}_j}{B}_{jk}{{\lambda}}_{j,t_j+1}+{{V}_j}^o \label{ASYN6}
	\end{align}
\end{subequations}
where $\varpi_j^a$ is the $j$-th component of $\bm{\varpi}^a$ and stepsizes $\eta, \alpha_{pq}, \alpha_\lambda >0$. $\tilde B_{jk}$ is the $j$th row and $k$th column element of matrix $\tilde{\textbf{B}}=\textbf{B}^2$. As $\textbf{B}$ has the same sparse structure with Laplacian of the subtree, the matrix $\textbf{B}^2$ has nonzero entries matching the neighbors and two-hop neighbors of each bus. This implies that each bus only needs the information of its neighbors and two-hop neighbors to compute the variable $\tilde\lambda_{j,t_j}$. The asynchronous distributed  voltage control (ASDVC) algorithm based on \eqref{inertia_algorithm_async} is given in Algorithm \ref{algorithm1}. 

\begin{algorithm}
	\caption{\textit{ASDVC}}
	\label{algorithm1}
	\begin{flushleft}
		\textbf{Input:} For bus $j$, the input is $(p_{j,0}, q_{j,0})\in\Omega_j$,  $\lambda_{j,0}\in \mathbb{R}$.
		
		%	\textbf{Output:} The 
		
		\textbf{Iteration at} \textit{$t_j$}: Suppose bus $j$'s clock ticks at time $t_j$, then bus $j$ is activated and updates its local variables as follows: 
		
		$\ \ $\textbf{Step 1: }\textbf{Reading phase}
		
		$\quad$ Get $\lambda_{k,t_k-\tau_{k}^{t_k}}, k\in N_j\cup N_j^2$ from its neighbors' and two-hop neighbors' output cache. 
		
		$\ \ $\textbf{Step 2:} \textbf{Computing phase}
		
		$\quad$ Calculate $\tilde p_{j,t_j}$, $\tilde q_{j,t_j}$ and $\tilde\lambda_{j,t_j}$ according to \eqref{ASYN1} and \eqref{ASYN2} respectively.
		
		$\quad$ Update ${\lambda}_{j,t_j+1}$, ${p}_{j,t_j+1}$, ${q}_{j,t_j+1}$ and ${{V}}_{j, t_j+1}$ according to \eqref{ASYN3} $-$ \eqref{ASYN6} respectively.
		
		$\ \ $\textbf{Step 3:} \textbf{Writing phase}
		
		$\quad$ Write ${\lambda}_{j,t_j+1}$ to its output cache and ${p}_{j,t_j+1}$, ${q}_{j,t_j+1}$, ${{V}}_{j, t_j+1}$ to its local storage. Increase $t_j$ to $t_j+1$.
	\end{flushleft}
\end{algorithm}
\begin{remark}[Asynchronous update]
	The main difference between this paper and \cite{liu2018hybrid} is that we design the asynchronous pattern for the partial primal-dual algorithm. It should be noted that this is not trivial. As proved in \cite{hale2017asynchronous}, the asynchronous distributed primal-dual algorithm cannot guarantee the convergence if dual variables are not updated simultaneously. In ASDVC, there is no need to update $\lambda$ simultaneously. To this end, we use neighbors' and two-hop neighbors' information. 
\end{remark}

\section{Optimality and Convergence}
In this section, we formulate the algorithm \eqref{inertia_algorithm_async} into a fixed-point iteration problem using operator-splitting method. Then,  its convergence and optimality of the equilibrium are proved.
\subsection{Algorithm Reformulation}
Define $z_j=\text{col}(p_j, q_j)$, $\tilde z_j=\text{col}(\tilde p_j, \tilde q_j)$. If the time delay is not considered, the compact form of \eqref{inertia_algorithm_async} can be obtained, denoted by SDVC. As $\textbf{V}$ is not in the iteration process, we omit it here.
\begin{subequations}\label{inertia_algorithm_compact_async} 
\setlength{\abovedisplayskip}{4pt}	
\setlength{\belowdisplayskip}{4pt}
	\begin{align}
	&\tilde{\textbf{z}}_{t}=\mathcal{P}_\Omega\left({\textbf{z}}_{t}-\alpha_{pq}(\nabla_{ z_t}g(z_t)-\text{col}(K{\bm{\lambda}}_{t}, {\bm{\lambda}}_{t})\right)\label{dis_algorithmc_async2}\\
	&\tilde{\bm{\lambda}}_{t}={\bm{\lambda}}_{t} + \alpha_\lambda\left(-\textbf{B}^2 {\bm{\lambda}}_{t} -2(K\cdot I_n,I_n)\tilde{\textbf{z}}_{t} + (K\cdot I_n,I_n){\textbf{z}}_t -\bm{\varpi}^a\right)\label{dis_algorithmc_async3}\\
	&\textbf{z}_{t+1}={\textbf{z}}_{t}+\eta(\tilde{\textbf{z}}_{t}-{\textbf{z}}_{t}) \label{dis_algorithmc_async5}\\
	&\bm{\lambda_{t+1}}=\bm{\lambda}_{t}+\eta(\tilde{\bm{\lambda}}_{t}-\bm{\lambda}_{t}) \label{dis_algorithmc_async4}
	\end{align}
\end{subequations}
In the rest of the paper, denote $F(\textbf{z})=\nabla_\textbf{z}g({\textbf{z}})$.
Equations \eqref{dis_algorithmc_async2}-\eqref{dis_algorithmc_async3} are equivalent to \footnote{The ``=" in \eqref{eqivalent1} is substituted by ``$\in$" in some literature. Here, we still use ``=" for the notation consistence if there is no confusion. }
\begin{subequations}\label{inertia_algorithm_compact_async2}
	\setlength{\abovedisplayskip}{4pt}	
	\setlength{\belowdisplayskip}{4pt}
	\begin{align}\label{eqivalent1}
	&-F({\textbf{z}}_t)=N_{\Omega}(\tilde{\textbf{z}}_{t})-\text{col}(K\tilde{\bm{\lambda}}_{t}, \tilde{\bm{\lambda}}_{t}) \nonumber\\
	&\quad\quad\quad\quad\quad\quad + \alpha_{pq}^{-1}(\tilde{\textbf{z}}_{t}-\textbf{z}_{t}) +(K\cdot I_n,I_n)^T(\tilde{\bm{\lambda}}_{t}-\bm{\lambda}_{t})\\
	\label{eqivalent2}
	&-\bm{\varpi}^a-\textbf{B}^2{\bm{\lambda}}_{t}=(K\cdot I_n,I_n)\tilde{\textbf{z}}_{t}+(K\cdot I_n,I_n) (\tilde{\textbf{z}}_{t}-\textbf{z}_{t})\nonumber\\
	&\quad\quad\quad\quad\quad\quad +\alpha_{\lambda}^{-1}(\tilde{\bm{\lambda}}_{t}-\bm{\lambda}_{t})
	\end{align}
\end{subequations}

Define following two operators
\begin{subequations}
\setlength{\abovedisplayskip}{4pt}	
\setlength{\belowdisplayskip}{4pt}
\begin{align}
\label{OperB2}
&\mathcal{C}:\left[ {\begin{array}{*{20}{c}}
	{\textbf{z}}\\
	{\bm{\lambda}}
	\end{array}} \right] \mapsto \left[ {\begin{array}{*{20}{c}}
	F({\textbf{z}})\\
	\bm{\varpi}^a+\textbf{B}^2{\bm{\lambda}}
	\end{array}} \right] \\
\label{OperU2}
&\mathcal{D}:\left[ {\begin{array}{*{20}{c}}
	{\textbf{z}}\\
	{\bm{\lambda}}
	\end{array}} \right] \mapsto \left[ {\begin{array}{*{20}{c}}
	N_{\Omega}(\textbf{z})-\text{col}(K\bm{\lambda}, \bm{\lambda})\\
	(K\cdot I_n,I_n)\textbf{z}
	\end{array}} \right]
\end{align}
\end{subequations}
and denote $\textbf{w}_t=\text{col}(\textbf{z}_t, \bm{\lambda}_t)$ and $\tilde{\textbf{w}}_t=\text{col}(\tilde{\textbf{z}_t},  \tilde{\bm{\lambda}}_t)$.

Then, \eqref{inertia_algorithm_compact_async2} can be rewritten as 
\begin{equation}
\setlength{\abovedisplayskip}{4pt}	
\setlength{\belowdisplayskip}{4pt}
-\mathcal{C}({\textbf{w}}_t)=\mathcal{D}(\tilde{\textbf{w}}_{t})+\Gamma\cdot(\tilde{\textbf{w}}_{t}-{\textbf{w}}_t)
\end{equation}
where 
\begin{equation}
\setlength{\abovedisplayskip}{4pt}	
\setlength{\belowdisplayskip}{4pt}
\Gamma:=\left[ {\begin{array}{*{20}{c}}
	\alpha_{pq}^{-1}I_{2n}&(K\cdot I_n,I_n)^T\\
	(K\cdot I_n,I_n) &\alpha_\lambda^{-1}I_n
	\end{array}} \right]
\end{equation}
Here, $ \alpha_{pq}, \alpha_\lambda$ are chosen to make $\Gamma$ is positive definite.

Denote the maximal and minimal eigenvalues of $\textbf{B}$ by $\sigma_{\max}$ and $\sigma_{\min}$ respectively. We have the following result.

\begin{lemma}\label{Lemma5}
	In terms of $\mathcal{C}$ and $\mathcal{D}$, we have following properties. 
	\begin{enumerate}[1)]
		\item Operator $\mathcal{C}$ is $\beta$-cocoercive under the 2-norm  with $0<\beta\le\min\{\frac{1}{\sigma_{\max}^2},\frac{1}{\vartheta}\}$; 
		\item Operator $\mathcal{D}$ is maximally monotone;
		\item $\Gamma^{-1}\mathcal{D}$ is maximally monotone under the $\Gamma$-induced norm \\$\|\cdot\|_{\Gamma}$;
		\item $({\rm{Id}}+\Gamma^{-1}\mathcal{D})^{-1}$ exists and is firmly nonexpansive.
	\end{enumerate}
\end{lemma}
\begin{proof}
	1): According to the definition of $\mathcal{C}$ and the definition of $\beta$-cocoercive, it suffice to prove that $\left\langle \mathcal{C}(\textbf{w}_1)-\mathcal{C}(\textbf{w}_2),\textbf{w}_1-\textbf{w}_2 \right\rangle \ge \beta \left\|\mathcal{C}(\textbf{w}_1)-\mathcal{C}(\textbf{w}_2) \right\|^2 $, or equivalently
	\begin{equation}\label{cocoercive_B}
	\setlength{\abovedisplayskip}{4pt}	
	\setlength{\belowdisplayskip}{4pt}
	\begin{split}
	\left(F(\textbf{z}_1)-F(\textbf{z}_2)\right)^T(\textbf{z}_1 - \textbf{z}_2) + ({\bm{\lambda}_1} - {\bm{\lambda}_2})^T\textbf{B}^2({\bm{\lambda}_1} - {\bm{\lambda}_2}) \ge \\
	\qquad\beta (\left\|\textbf{B}^2{\bm{\lambda}_1} - \textbf{B}^2 {\bm{\lambda}_2}\right\|^2+\left\|F(\textbf{z}_1)-F(\textbf{z}_2)\right\|^2)
	\end{split}
	\end{equation}
	
	Notice that $\bm{\varpi}^a+\textbf{B}^2{\bm{\lambda}}$ is the gradient of function $\hat f({\bm{\lambda}})=\frac{1}{2}{\bm{\lambda}}^T\textbf{B}^2{\bm{\lambda}}+{\bm{\lambda}}^T\bm{\varpi}^a$. As $\nabla^2\hat f({\bm{\lambda}})=\textbf{B}^2>0$, $\hat f({\bm{\lambda}})$ is a convex function. For its gradient, we have 
	\begin{equation}
	\setlength{\abovedisplayskip}{4pt}	
	\setlength{\belowdisplayskip}{4pt}
	\left\|\textbf{B}^2({\bm{\lambda}_1} - {\bm{\lambda}_2})\right\| \le \left\|\textbf{B}^2\right\|\left\|{\bm{\lambda}_1} - {\bm{\lambda}_2}\right\| = \sigma_{\max}^2\left\|{\bm{\lambda}_1} - {\bm{\lambda}_2}\right\|
	\end{equation}
	Thus, $\nabla\hat f({\bm{\lambda}})$ is $\sigma_{\max}^2$-Lipschitzian. Then, 
	$\nabla\hat f({\bm{\lambda}})= \bm{\varpi}^a+\textbf{B}^2{\bm{\lambda}}$ is $\frac{1}{\sigma_{\max}^2}$-cocoercive \cite[Corollary 18.16]{bauschke2011convex}, i.e.,
	\begin{equation}\label{cocoercive_C}
	\setlength{\abovedisplayskip}{4pt}	
	\setlength{\belowdisplayskip}{4pt}
	({\bm{\lambda}_1} - {\bm{\lambda}_2})^T\textbf{B}^2({\bm{\lambda}_1} - {\bm{\lambda}_2}) \ge \frac{1}{\sigma_{\max}^2} \left\|\textbf{B}^2{\bm{\lambda}_1} - \textbf{B}^2 {\bm{\lambda}_2}\right\|^2
	\end{equation}

	Moreover, since $F$ is $\frac{1}{\vartheta}$-cocoercive, i.e.,
	\begin{equation}\label{cocoercive_F}
	\setlength{\abovedisplayskip}{4pt}	
	\setlength{\belowdisplayskip}{4pt}
	\left(F(\textbf{z}_1)-F(\textbf{z}_2)\right)^T(\textbf{z}_1 - \textbf{z}_2)\ge \frac{1}{\vartheta}\left\|F(\textbf{z}_1)-F(\textbf{z}_2)\right\|^2
	\end{equation}
	Combining \eqref{cocoercive_C}, \eqref{cocoercive_F} and taking $0<\beta\le\min\{\frac{1}{\sigma_{\max}^2},\frac{1}{\vartheta}\}$, we can get the first assertion. 
	
	2): The operator $\mathcal{D}$ can be rewritten as 
	\begin{eqnarray}
	\setlength{\abovedisplayskip}{3pt}	
	\setlength{\belowdisplayskip}{3pt}
	\begin{split}
		\mathcal{D}&=\left[ {\begin{array}{*{20}{c}}
		0& -(K I_n,I_n)^T\\
		(K I_n,I_n) &0
		\end{array}} \right]\left[ {\begin{array}{*{20}{c}}
		{\textbf{z}}\\
		{\bm{\lambda}}
		\end{array}} \right]+\left[ {\begin{array}{*{20}{c}}
		N_{\Omega}(\textbf{z})\\
		0
		\end{array}} \right]\\
	&=\mathcal{D}_1+\mathcal{D}_2
	\end{split}
	\end{eqnarray}
	As $\mathcal{D}_1$ is a skew-symmetric matrix, $\mathcal{D}_1$ is maximally monotone \cite[Example 20.30]{bauschke2011convex}. Moreover, $N_{\Omega}(\textbf{z})$ and $0$ are all maximally monotone \cite[Example 20.41]{bauschke2011convex}, so $\mathcal{D}_2$ is also maximally monotone. Thus, $\mathcal{D}=\mathcal{D}_1+\mathcal{D}_2$ is maximally monotone.
	
	3) As $\Gamma$ is symmetric positive definite and $\mathcal{D}$ is maximally monotone, we can prove that $\Gamma^{-1}\mathcal{D}$ is maximally monotone by the similar analysis in Lemma 5.6 of \cite{yi2017distributed}.  
	
	4) As $\Gamma^{-1}\mathcal{D}$ is maximally monotone, $({\rm{Id}}+\Gamma^{-1}\mathcal{D})^{-1}$ exists and is firmly nonexpansive by \cite[Proposition 23.7]{bauschke2011convex}. 
\end{proof}

By the last assertion of Lemma \ref{Lemma5}, \eqref{inertia_algorithm_compact_async} is equivalent to 
\begin{subequations}\label{Opera2}
	\setlength{\abovedisplayskip}{4pt}	
	\setlength{\belowdisplayskip}{4pt}
	\begin{align}
	\tilde{\textbf{w}}_t&=({\rm{Id}}+\Gamma^{-1}\mathcal{D})^{-1}({\rm{Id}}-\Gamma^{-1}\mathcal{C}){\textbf{w}}_t \\
	\textbf{w}_{t+1}&={\textbf{w}}_t+\eta(\tilde{\textbf{w}}_t-{\textbf{w}}_{t})
	\end{align}
\end{subequations}

Denote $\mathcal{S}_1=({\rm{Id}}+\Gamma^{-1}\mathcal{D})^{-1}$, $\mathcal{S}_2=({\rm{Id}}-\Gamma^{-1}\mathcal{C})$ and $\mathcal{S}=\mathcal{S}_1\mathcal{S}_2$, and then we have following results.
\begin{lemma}\label{Lemma6}
	Take $0<\beta\le\min\{\frac{1}{\sigma_{\max}^2},\frac{1}{\vartheta}\}$, $\kappa>\frac{1}{2\beta}$, and the step sizes  $\alpha_{pq}, \alpha_{\lambda}$ such that $\Gamma-\kappa I$ is positive semi-definite. Following results are true under the $\Gamma$-induced norm $\|\cdot\|_{\Gamma}$.
	\begin{enumerate}[1)]
		\item $\mathcal{S}_1$ is a $\frac{1}{2}$-averaged operator, i.e., $\mathcal{S}_1 \in \mathcal{A}\left(\frac{1}{2}\right)$;
		\item $\mathcal{S}_2$ is a $\frac{1}{2\beta\kappa}$-averaged operator, i.e., $\mathcal{S}_2 \in \mathcal{A}\left(\frac{1}{2\beta\kappa}\right)$;
		\item $\mathcal{S}$ is a $\frac{2\kappa\beta}{4\kappa \beta-1}$-averaged operator, i.e., $\mathcal{S} \in \mathcal{A}\left(\frac{2\kappa\beta}{4\kappa \beta-1}\right)$. 
	\end{enumerate} 
\end{lemma}
\begin{proof}
	1): From the assertion 4) of Lemma \ref{Lemma5}, $\mathcal{S}_1=({\rm{Id}}+\Gamma^{-1}\mathcal{D})^{-1}$ is firmly nonexpansive, implying $\mathcal{S}_1\in\mathcal{A}(\frac{1}{2})$.
	
	2): First, we prove that $\Gamma^{-1}\mathcal{C}$ is $\beta\kappa$-cocoercive, i.e., 
	\begin{equation}\label{cocoercive1}
	\setlength{\abovedisplayskip}{4pt}	
	\setlength{\belowdisplayskip}{4pt}
	\begin{split}
	\left\langle\Gamma^{-1}\mathcal{C}(\textbf{w}_1)-\Gamma^{-1}\mathcal{C}(\textbf{w}_2), \textbf{w}_1-\textbf{w}_2\right\rangle_{\Gamma}\ge\\
	\qquad\beta\kappa\left\|\Gamma^{-1}\mathcal{C}(\textbf{w}_1)-\Gamma^{-1}\mathcal{C}(\textbf{w}_2)\right\|^2_{\Gamma}
	\end{split}
	\end{equation}
	Denote the maximal and minimal eigenvalues of $\Gamma$ by $\delta_{\max}$ and $\delta_{\min}$ respectively, and we have $\delta_{\max}\ge\delta_{\min}\ge\kappa>0$. Moreover, the Euclidean norms of $\Gamma$ and $\Gamma^{-1}$ are $\left\|\Gamma\right\|_2=\delta_{\max}$ and $\left\|\Gamma^{-1}\right\|_2=\frac{1}{\delta_{\min}}$ \cite[Proposition 5.2.7, 5.2.8]{meyer2000matrix}. 
	
	For the right hand side of \eqref{cocoercive1}, we have
	\begin{align}\label{cocoercive2}
	&\beta\kappa\left\|\Gamma^{-1}\mathcal{C}(\textbf{w}_1)-\Gamma^{-1}\mathcal{C}(\textbf{w}_2)\right\|^2_{\Gamma}=\beta\kappa\left\|\mathcal{C}(\textbf{w}_1)-\mathcal{C}(\textbf{w}_2)\right\|^2_{\Gamma^{-1}}\nonumber\\
	&\qquad=\beta\kappa(\mathcal{C}(\textbf{w}_1)-\mathcal{C}(\textbf{w}_2))^T\Gamma^{-1}(\mathcal{C}(\textbf{w}_1)-\mathcal{C}(\textbf{w}_2))\nonumber\\
	&\qquad\le\beta\kappa\left\|\Gamma^{-1}\right\|_2\left\|\mathcal{C}(\textbf{w}_1)-\mathcal{C}(\textbf{w}_2)\right\|^2_{2}\nonumber\\
	&\qquad\le\beta\kappa\cdot\frac{1}{\kappa}\left\|\mathcal{C}(\textbf{w}_1)-\mathcal{C}(\textbf{w}_2)\right\|^2_{2}
	\end{align}
	where the first "$\le$" is due to the Cauchy-Schwarz  inequality and the second is due to $\delta_{\min}\ge\kappa$.
	
	For the left part of \eqref{cocoercive1}, we have
	\begin{eqnarray}\label{cocoercive3}
	\setlength{\abovedisplayskip}{4pt}	
	\setlength{\belowdisplayskip}{4pt}
	&\left\langle\Gamma^{-1}\mathcal{C}(\textbf{w}_1)-\Gamma^{-1}\mathcal{C}(\textbf{w}_2), \textbf{w}_1-\textbf{w}_2\right\rangle_{\Gamma}\nonumber\\
	&\qquad\qquad\qquad=\left\langle \mathcal{C}(\textbf{w}_1)- \mathcal{C}(\textbf{w}_2), \textbf{w}_1-\textbf{w}_2\right\rangle\nonumber\\
	&\qquad\qquad\qquad\ge \beta \left\|\mathcal{C}(\textbf{w}_1)-\mathcal{C}(\textbf{w}_2) \right\|^2 _2
	\end{eqnarray}
	where the inequality is from assertion 1) of Lemma \ref{Lemma5}. From \eqref{cocoercive2} and \eqref{cocoercive3}, we have \eqref{cocoercive1}.
	
	As $\Gamma^{-1}\mathcal{C}$ is $\beta\kappa$-cocoercive, we have $\beta\kappa\Gamma^{-1}\mathcal{C}\in\mathcal{A}(\frac{1}{2})$. That is to say, there is a nonexpansive operator $\tilde{\mathcal{S}}$ such that $\beta\kappa\Gamma^{-1}\mathcal{C}=\frac{1}{2}{\rm{Id}}+\frac{1}{2}\tilde{\mathcal{S}}$, i.e., $\Gamma^{-1}\mathcal{C}=\frac{1}{2\beta\kappa}{\rm{Id}}+\frac{1}{2\beta\kappa}\tilde{\mathcal{S}}$. Then, 
	\begin{equation}
	\setlength{\abovedisplayskip}{4pt}	
	\setlength{\belowdisplayskip}{4pt}
	\mathcal{S}_2={\rm{Id}}-\Gamma^{-1}\mathcal{C}=\left(1-\frac{1}{2\beta\kappa}\right){\rm{Id}}-\frac{1}{2\beta\kappa}\tilde{\mathcal{S}} 
	\end{equation}
	As $0<\frac{1}{2\beta\kappa}<1$ and $-\tilde{\mathcal{S}}$ is also nonexpansive, we have $\mathcal{S}_2\in \mathcal{A}(\frac{1}{2\beta\kappa})$.

	3): From \cite[Propsition 2.4]{combettes2015compositions}, $\mathcal{S}=\mathcal{S}_1\mathcal{S}_2$ is a $a$-averaged operator with $a=\frac{a_1+a_2-2a_1a_2}{1-a_1a_2}$, if $\mathcal{S}_1$ is $a_1$-averaged and $\mathcal{S}_2$ is $a_2$-averaged. As $\mathcal{S}_1 \in \mathcal{A}\left(\frac{1}{2\beta\kappa}\right)$ and $\mathcal{S}_2 \in \mathcal{A}\left(\frac{1}{2}\right)$, we have $\mathcal{S} \in \mathcal{A}\left(\frac{2\kappa\beta}{4\kappa \beta-1}\right)$.	
\end{proof}

By the definition of the averaged operator and assertion 3) of Lemma \ref{Lemma6}, there exists a nonexpansive operator $\mathcal{T}$ such that 
\begin{equation}\label{nonexpansive}
\setlength{\abovedisplayskip}{4pt}	
\setlength{\belowdisplayskip}{4pt}
\mathcal{S}=\left(1-\frac{2\kappa\beta}{4\kappa \beta-1}\right){\rm{Id}}+\frac{2\kappa\beta}{4\kappa \beta-1}\mathcal{T}
\end{equation}
Apparently, operators $\mathcal{S}$ and $\mathcal{T}$ have the same fixed points, i.e., $Fix (\mathcal{S})=Fix (\mathcal{T})$.

We convert the asynchronous algorithm into a fixed-point iteration problem with an averaged operator. Moreover, we also construct a nonexpansive operator $\mathcal{T}$, which enables us to prove the convergence of the asynchronous algorithm ASDVC. 

%\subsection{Optimality}
%
%
%Considering dynamic system \eqref{inertia_algorithm_async}, we give the following  definition of its equilibrium point.
%\begin{definition}\label{equilibrium2}
%	A point $\textbf{w}^*= (\textbf{V}^*, \textbf{z}^*, \bm{\lambda}^*)$ 
%	is an {equilibrium point} of  system \eqref{inertia_algorithm_async} if $\lim_{t\rightarrow +\infty} \textbf{w}_{t}=\textbf{w}^*$ holds.
%\end{definition}
%
%
%Then, we have the following result.
%\begin{theorem}
%	\label{optimality2}
%	The component $\textbf{V}^*, \textbf{z}^*, \bm{\lambda}^*$ of the equilibrium point $w^*$ satisfies the KKT condition \eqref{KKT}, i.e., it is the primal-dual optimal solution to the optimization problem \eqref{eq_opt2}.
%\end{theorem}
%\begin{proof}
%	The proof is so easy.
%\end{proof}

\subsection{Optimality of the equilibrium point}

The definition of the equilibrium point of ASDVC is introduced as follows.
\begin{definition}\label{equilibrium2}
	A point $\textbf{w}^*=\text{col}(w^*_j)= \text{col}({x}_j^*, {\lambda}_j^*)$ is an {equilibrium point} of  system \eqref{inertia_algorithm_async} if $\lim_{t_j\rightarrow \infty} w_{t_j}=w_j^*,\ \forall j$ holds.
\end{definition}

Now, we give the KKT condition of the optimization problem \eqref{eq_opt2} \cite[Theorem 3.25]{ruszczynski2006nonlinear}.
\begin{subequations}
	\label{KKT}
	\setlength{\abovedisplayskip}{4pt}	
	\setlength{\belowdisplayskip}{4pt}
	\begin{align}
	0&=( {\textbf{V}}-{\textbf{V}}^o)+\textbf{B}^T\bm{\lambda} \\
	0&=\nabla_{{\textbf{z}}}g({\textbf{z}})-\text{col}(K\bm{\lambda}, \bm{\lambda}) +N_{\Omega}(\textbf{z})\\
	0&=\bm{\varpi}^s-\textbf{B} {\textbf{V}}+(K\cdot I,I)\textbf{z}
	\end{align}
\end{subequations}

Denote $ \textbf{V}^*=-\textbf{B}^T{\bm{\lambda}}^*+{\textbf{V}}^o $, and we have the following result.
\begin{theorem}
	\label{optimality}
	The point $(\textbf{V}^*, \textbf{z}^*, \bm{\lambda}^*)$ satisfies the KKT condition \eqref{KKT}, i.e., it is the primal-dual optimal solution to the optimization problem \eqref{eq_opt2}.
\end{theorem}
\begin{proof}
	By Definition \ref{equilibrium2}, we know ${w}_j^*=\lim\limits_{t_j\rightarrow \infty} {w}_{j,t_j-1}=\lim\limits_{t_j\rightarrow \infty} {w}_{j,t_j}=$ $\lim\limits_{t_j\rightarrow \infty} {w}_{j,t_j+1}=$  $\lim\limits_{t_j\rightarrow \infty} \tilde{{w}}_{j,t_j}$. From \eqref{inertia_algorithm_compact_async2}, we have 
	\begin{subequations}\label{equilibrium_compact}
		\setlength{\abovedisplayskip}{4pt}	
		\setlength{\belowdisplayskip}{4pt}
		\begin{align}
		&-( {\textbf{V}}^*-{\textbf{V}}^o)=\textbf{B}^T\bm{\lambda}^*\\
		&-\nabla_{{\textbf{z}^*}}g({\textbf{z}^*})=N_{\Omega}(\textbf{z}^*)-\text{col}(K\bm{\lambda}^*, \bm{\lambda}^*) \\
		&-\bm{\varpi}^s=-\textbf{B} {\textbf{V}}^*+(K\cdot I,I)\textbf{z}^*
		\end{align}
	\end{subequations}
	Comparing \eqref{KKT} and \eqref{equilibrium_compact}, we know $(\textbf{V}^*, \textbf{z}^*, \bm{\lambda}^*)$ satisfies the KKT condition. This completes the proof. 
\end{proof}

\subsection{Convergence analysis}
In this subsection, we investigate the convergence of ASDVC. We first treat ASDVC as a randomized block-coordinate fixed-point iteration problem with delayed information. Then, the results in \cite{peng2016arock} can be applied. 

To prove the convergence of ASDVC, we need introduce a global clock to substitute the local clocks of individual buses in ASDVC. The main idea is to queue $t_j$ of all buses in the order of  real time, and use a new number $t$ to denote the $t$-th iteration in the queue. Take two local clocks as an example. Suppose the local clocks to be $t_1=\{1, 3, 5, \cdots \}$ and $t_2=\{2, 4, 6, \cdots \}$, and then the global clock is $t=\{1, 2, 3, 4, 5, 6, \cdots \}$.
In the global clock, it is assumed that the probability bus $j$ is activated to update its local variables follows a uniform distribution. Hence, each bus is activated with the same probability. Note that the global clock is only used for convergence analysis, but it does not exist in the application. 

%
%\begin{figure}[htp]
%	\centering
%	\includegraphics[width=0.4\textwidth]{merge_k}
%	\caption{Local clocks versus global clock}
%	\label{merge_k}
%\end{figure}

Define vectors $\psi_j\in \mathbb{R}^{3n}, j\in \mathcal{N}$. The $i$th entry of $\psi_j$ is denoted by $[\psi_j]_i$. Define $[\psi_j]_i=1$ if the $i$th coordinate of $\textbf{w}$ is also a coordinate of $w_j$, and $[\psi_j]_i=0$, otherwise. Denote by $\xi$ a random variable (vector) taking values in $\psi_j, j\in \mathcal{N}$. Then $\textbf{Prob}(\xi=\psi_j)=1/n$ also follows a uniform distribution. Let $\xi_t$  be the value of $\xi$ at the $t$th iteration. Then, a randomized block-coordinate fixed-point iteration for \eqref{Opera2} is given by
\begin{equation}
\setlength{\abovedisplayskip}{4pt}	
\setlength{\belowdisplayskip}{4pt}
\label{random_fixed}
\textbf{w}_{t+1}=\textbf{w}_t+\eta\xi_t\circ(\mathcal{S}(\textbf{w}_t)-\textbf{w}_t)
\end{equation}
where $\circ$ denotes the Hadamard product. In \eqref{random_fixed}, only one bus $j$ is activated at each iteration. 
% without loss of generality\footnote{Note that this model helps formulate the algorithm and analyze its convergence. In implementation, we allow two or more buses are activated simultaneously, which can be modeled as two or more iterations in analysis.  }. 

Since  \eqref{random_fixed} is delay-free,  we further modify it for considering delayed information, which is 
\begin{equation}
\label{delay_random_fixed}
\setlength{\abovedisplayskip}{4pt}	
\setlength{\belowdisplayskip}{4pt}
\textbf{w}_{t+1}=\textbf{w}_t+\eta\xi_t\circ(\mathcal{S}(\hat {\textbf{w}}_t)- \textbf{w}_t)
\end{equation}
where $\hat {\textbf{w}}_t$ is the information with delay at iteration $t$. 
We will show that Algorithm \ref{algorithm1} can be written as \eqref{delay_random_fixed} if $\hat {\textbf{w}}_t$ is properly defined.
Suppose bus $j$ is activated at the iteration $t$, then $\hat {\textbf{w}}_t$ is defined as follows. For bus $j$, replace $p_{j,t_j}$, $q_{j,t_j}$ and $\lambda_{j,t_j}$ with $p_{j,t_j-\tau_{j}^{t_j}}$, $q_{j,t_j-\tau_{j}^{t_j}}$ and $\lambda_{j,t_j-\tau_{j}^{t_j}}$. Similarly, replace $\lambda_{k,t_k}$ with $\lambda_{k,t_k-\tau_{k}^{t_k}}$ from its neighbors and two-hop neighbors. For inactivated buses, their state values keep unchanged. 
%Then we can get \eqref{delay_random_fixed}.

Before proving the convergence, we make an assumption.
\begin{assumption}
	\label{bounded_delay}
	The maximal time delay between two consecutive iterations is bounded by $\chi$, i.e., $\{\max\{\tau_j^t\}\}\le \chi, \forall t, \forall j$.
\end{assumption}

With the assumption, we have the convergence result.
\begin{theorem}\label{convergence}
	Suppose Assumptions \ref{Lipschitz}, \ref{bounded_delay} holds. Take $0<\beta\le\min\{\frac{1}{\sigma_{\max}^2},\frac{1}{\vartheta}\}$, $\kappa>\frac{1}{2\beta}$, and the step sizes  $\alpha_{pq}, \alpha_{\lambda}$ such that $\Gamma-\kappa I$ is positive semi-definite. Choose $0<\eta< \frac{1}{1+2\chi/\sqrt{n}} \frac{4\kappa \beta-1}{2\kappa\beta}$. Then, with ASDVC, $\textbf{w}_t$ converges to the point $\textbf{w}^{*}$ defined in Definition \ref{equilibrium2} with probability 1.
\end{theorem}
\begin{proof}
	Combining \eqref{nonexpansive} and \eqref{delay_random_fixed}, we have 
	\begin{align}
	\textbf{w}_{t+1}&=\textbf{w}_t+\eta\xi_t\circ\left(\left(1-\frac{2\kappa\beta}{4\kappa \beta-1}\right)\hat {\textbf{w}}_t - \textbf{w}_t+\frac{2\kappa\beta}{4\kappa \beta-1}\mathcal{T}(\hat {\textbf{w}}_t)\right) \nonumber\\
	\label{nonexpansive2}
	&=\textbf{w}_t+\eta\xi_t\circ\left(\hat {\textbf{w}}_t- \textbf{w}_t	+\frac{2\kappa\beta}{4\kappa \beta-1}(\mathcal{T}(\hat {\textbf{w}}_t)-\hat {\textbf{w}}_t)\right)
	\end{align}
	With $w_{j,t-\tau_{j}^t}=w_{j,t-\tau_{j}^t+1}=,\cdots,=w_{j,t}$, we have $\xi_k\circ\left(\hat {\textbf{w}}_t- \textbf{w}_t\right)=0$. Thus, \eqref{nonexpansive2} is equivalent to 
	\begin{equation}
	\label{nonexpansive3}
	\setlength{\abovedisplayskip}{4pt}	
	\setlength{\belowdisplayskip}{4pt}
	\textbf{w}_{t+1}=\textbf{w}_t+\frac{2\eta\kappa\beta}{4\kappa \beta-1}\xi_t\circ\left(\mathcal{T}(\hat {\textbf{w}}_t)-\hat {\textbf{w}}_t\right)
	\end{equation}
	In fact, \eqref{nonexpansive3} has the form of the ARock algorithms proposed in \cite{peng2016arock}.
	In Lemma 13 and Theorem 14 of \cite{peng2016arock}, it is proved that $\textbf{w}_t$ generated by \eqref{nonexpansive3} is bounded. Moreover, if $\eta$ satisfies $0<\eta< \frac{1}{1+2\chi/\sqrt{n}} \frac{4\kappa \beta-1}{2\kappa\beta}$, $\textbf{w}_t$ converges to a random variable that takes value in the fixed points of $\mathcal{T}$ with probability 1, denoted by $\textbf{w}^*$. Recall $Fix (\mathcal{S})=Fix (\mathcal{T})$ and Theorem  \ref{optimality}, and we know that $\textbf{w}^*$ satisfies the KKT condition \eqref{KKT}, i.e., the equilibrium in Definition \ref{equilibrium2}. 
	This completes the proof.
\end{proof}

\begin{remark}[Nonexpansive operator]
	As $\mathcal{S}$ is $\frac{2\kappa\beta}{4\kappa \beta-1}$-averaged, it is also a nonexpansive operator \cite[Remark 4.24]{bauschke2011convex}. Then, in Theorem \ref{convergence}, we can also use $\mathcal{S}$ instead of $\mathcal{T}$ to prove the convergence of ASDVC. In this situation, the bound of $\eta$ is $0<\eta< \frac{1}{1+2\chi/\sqrt{n}}$. Since $\kappa>\frac{1}{2\beta}$, we have $\frac{4\kappa \beta-1}{2\kappa\beta}>1$. This implies that the operator $\mathcal{T}$ can increase the upper bound of $\eta$ compared with $\mathcal{S}$. 
\end{remark}

\section{Implementation}
\subsection{Communication graph}
	Although two-hop neighbors' information is utilized in the algorithm \eqref{inertia_algorithm_async}, the communication graph still can be fully distributed, i.e., neighboring communication. In \eqref{ASYN2}, two-hop neighbors' information is needed to obtain $\tilde B_{jk}$ and $\lambda_{k,t_k-\tau_{k}^{t_k}},  k\in N^2_j$. For $\tilde B_{jk}$, it can be obtained from twice neighboring communications. In addition, as the topology of a distribution network does not change frequently, $\tilde B_{jk}$ can be obtained in advance. For $\lambda_{k,t_k-\tau_{k}^{t_k}},  k\in N^2_j$, it also can be obtained from neighboring communications, which is illustrated in Fig.\ref{two_step}.  
	\begin{figure}[t]
		\setlength{\abovecaptionskip}{0pt}
		\centering
		\includegraphics[width=0.25\textwidth]{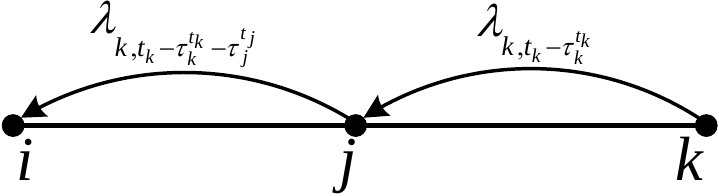}
		\caption{Two-step communications}
		\label{two_step}
	\end{figure}

	Node $i$ can get the information of node $k$ by twice communications. In this situation, the time delay may be longer. As this is the asynchronous algorithm, we treat $\tau_{j}^{t_j}+\tau_{k}^{t_k}$ as one delay $\tau_{k}^{t_k}$ if there is no confusion. Then, one step and two-step communication delays can be formulated into a uniform framework. In this way, only neighboring communications are needed in the asynchronous algorithm.

\subsection{Online Implementation}\label{Online_implementation}
	
	In the ASDVC, $\varpi_j^a$ is assumed to be available for every bus $ j $. By its definition, $\varpi_j^a$ is determined by almost all of the power injection in the network, which implies that a centralized method is needed to get their values. Thus, if the system states vary rapidly with the variation of renewable generations and loads, they are difficult to obtain. From \eqref{eq_opt2b}, $\bm{\varpi}^s$ can be obtained from an equivalent way if we can measure the local voltage, active and reactive power injections and get the neighbors' voltages. 
	Similar is $ \bm{\varpi}^a $ due to $\bm{\varpi}^a=\bm{\varpi}^s-\textbf{B}{\textbf{V}}^o$. Denote the set of buses connected directly to bus $0$ by $N_0$. If we set ${\textbf{V}}^o=0.5\times\textbf{1}$, we have $\textbf{B}{\textbf{V}}^o=\left\{  \begin{array}{l}
	0,\ \ \  j\notin N_0\\
	\frac{1}{2x_{0j}}, j\in N_0
	\end{array} \right.$.
	Then, $\varpi_j^a$ in the ASDVC algorithm can be obtained by  
	\begin{equation}\label{ASDVC_online}
	\setlength{\abovedisplayskip}{4pt}	
	\setlength{\belowdisplayskip}{4pt}
	\varpi_j^a=\left\{  \begin{array}{l}
	-\sum\nolimits_{k\in N_j} B_{jk}  V_k^m
	+K{p_j^m}+{q}^m_j, j\notin N_0\\
	-\frac{1}{2x_{0j}}-\sum\nolimits_{k\in N_j} B_{jk}  V_k^m
	+K{p_j^m}+{q}^m_j, j\in N_0
	\end{array} \right.	
	\end{equation}
	where ${p_j^m}, {q}^m_j$ are measured active and reactive power locally and $V_k^m$ is the square of  measured voltage of the neighbor. 
	
	In \eqref{ASDVC_online}, only communications between neighbors are needed. We can also use $p_{j,t}, q_{j,t}$ instead of ${p_j^m}, {q}^m_j$ to avoid power measurements. In the inverter integrated DERs, $p_{j,t}\approx{p_j^m}$ and $q_{j,t}\approx{q}^m_j$ as the response is very fast. The voltage measurements contain the latest system information, which makes the implementation track the time-varying operating conditions.

\section{Case Studies}\label{sec:sr}
In this section, simulation results are presented to demonstrate the effectiveness of the proposed voltage control methods. To this end, an 8-bus feeder and the IEEE 123-bus feeder are utilized as test systems. Each bus is equipped with a certain amount of PVs, which are able to offer flexible active and reactive power supplies to the feeder. Some buses have other controllable DERs like small hydro plants and energy storage systems. The simulations are implemented in Matlab R2017b simulator, and the OpenDSS is used for solving the ac power flow. 

\subsection{8-bus feeder}
%is set as $f_i=\frac{1}{2}a_i(P_i^g)^2+b_iP_i^g$
An 8-bus distribution network is considered \cite{Tang2018Fast}, whose diagram is shown in Fig.\ref{8bus}. The impedance of each line segment is identical, which is $0.9216+j0.4608\ \Omega$ with $K = 2$. All buses except bus 0 are equipped with DERs with active power limit as $\overline{\textbf{p}}=-\underline{\textbf{p}}=[90, 100, 0,$ $  120, 170, 90, 70]$kW, reactive power limit as $\overline{\textbf{q}}=-\underline{\textbf{q}}=[100, $ $   100, 110, 100, 130, 100, 120]$kVar. The capacity limit of inverter at each bus is $0.9*\sqrt{\overline{\textbf{p}}^2+\overline{\textbf{q}}^2}$kVA. 
\begin{figure}[t]
	\setlength{\abovecaptionskip}{0pt}
	\centering
	\includegraphics[width=0.15\textwidth]{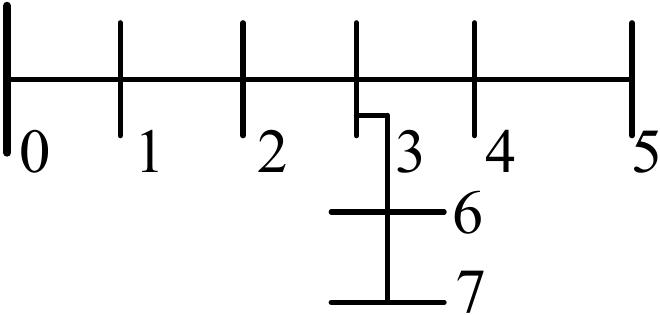}
	\caption{The graph of the 8-bus distribution network}
	\label{8bus}
\end{figure}

We first use CPLEX to obtain the optimal solution, which is $ \textbf{V}^*=[0.9934,1.0063,$  $1.0083,$  $1.0282,$  $0.9492,1.0073,0.9987] $, $ \textbf{p}^*=[63.08, 65.41, 0, 120, 170,$  $ 
70.57,$  $ 70] $kW, $ \textbf{q}^*=[31.53, $  $32.71,$  $ 63, 73.24, 90.54, 35.28, 41.29] $kVar. Then, SDVC and ASDVC are utilized in the voltage control in 8-bus feeder. 
We compare the controllers' performance by showing how $ \frac{\|\textbf{w}-\textbf{w}^*\|_2^2}{\|\textbf{w}^*\|_2^2} $ is evolving with the number of average iterations of each MG, which is given in  Fig.\ref{voltage_error_8_node}. 
\begin{figure}[t]
	\setlength{\abovecaptionskip}{0pt}
	\centering
	\includegraphics[width=0.38\textwidth]{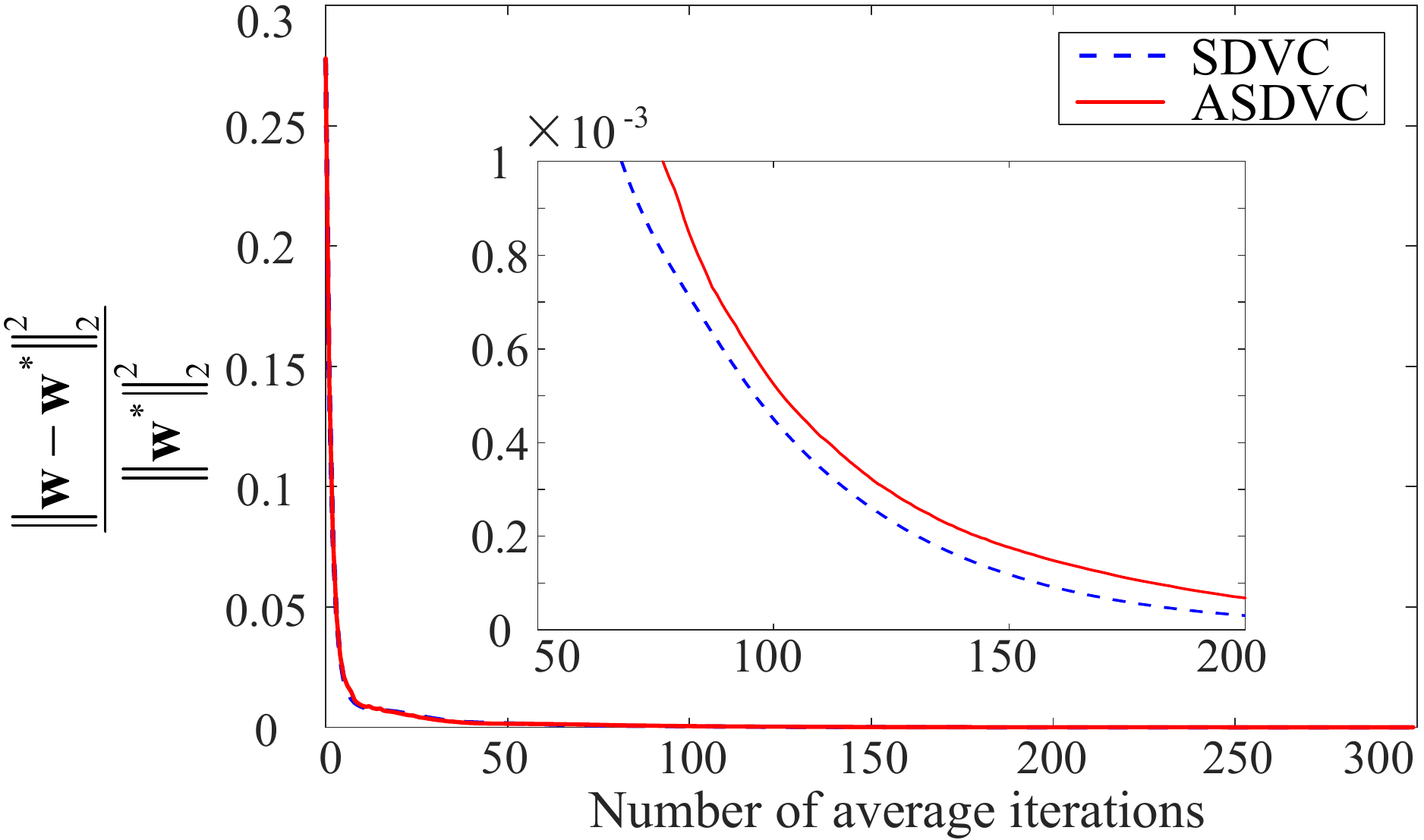}
	\caption{Comparison of algorithm convergence in terms of number of average iterations for SDVC and ASDVC. In the ASDVC, the random delay between $ 0\sim10 $ iterations is considered. }
	\label{voltage_error_8_node}
\end{figure}
The SDVC and ASDVC algorithms have similar convergence speed, taking about 50 iterations. In contrast, ASDVC is only a bit of inferior to SDVC in terms of the number of average iterations.

\subsection{IEEE 123-bus feeder}
%\subsubsection{Minute-sampled data}
In addition to the 8-bus feeder, we also test the proposed method on the IEEE 123-bus system to show the scalability and practicability, the diagram of which is shown in Fig.\ref{IEEE123}. It should be noted that the IEEE 123-bus system is not homogeneous, where the $r_{ij}/x_{ij}$ ranges from $0.42\sim2.02$. In the simulation, we take $K=1$, which also shows the robustness of our method. In this case, the real data of residential load and solar generation is utilized. The minute-sampled profiles of active and reactive load are from an online data repository \cite{website_load}, and we use the data of July 13th, 2010. The minute-sampled profile of solar generation is from  \cite{website_PV}, which were collected in a city in Utah, U.S. and we use the data of July 14th, 2010. The profiles of active, reactive loads and solar generation are given in Fig.\ref{load_demand}, where the black curve is the active power (kW), red curve is the reactive power (kVar) and dotted blue line is the solar generation (kW). In the simulation, the tap positions of voltage regulators are held constant in order to better capture the performance of proposed method. The voltage at the substation of the feeder head is set as 1 p.u. and the value of $\textbf{V}^o$ is $0.5\times\textbf{1}$ p.u.. Each residential home is equipped with a solar generation. The capacity limit at each bus is $20$kVA. The upper limit of active power is the instantaneous generation of the PV and the reactive power limit is determined correspondingly. Some buses are equipped with small hydro plants (marked as red in Fig.\ref{IEEE123}). The active power limits are $300$kW.
When load and solar generation change, the method in Section \ref{Online_implementation} is utilized for the online implementation. In each minute, a quasi-static operating condition is adopted, and the proposed controller is implemented with each iteration updated every 0.2 seconds (a total of 300 iterations per minute). 

\begin{figure}[t]
	\setlength{\abovecaptionskip}{0pt}
	\centering
	\includegraphics[width=0.4\textwidth]{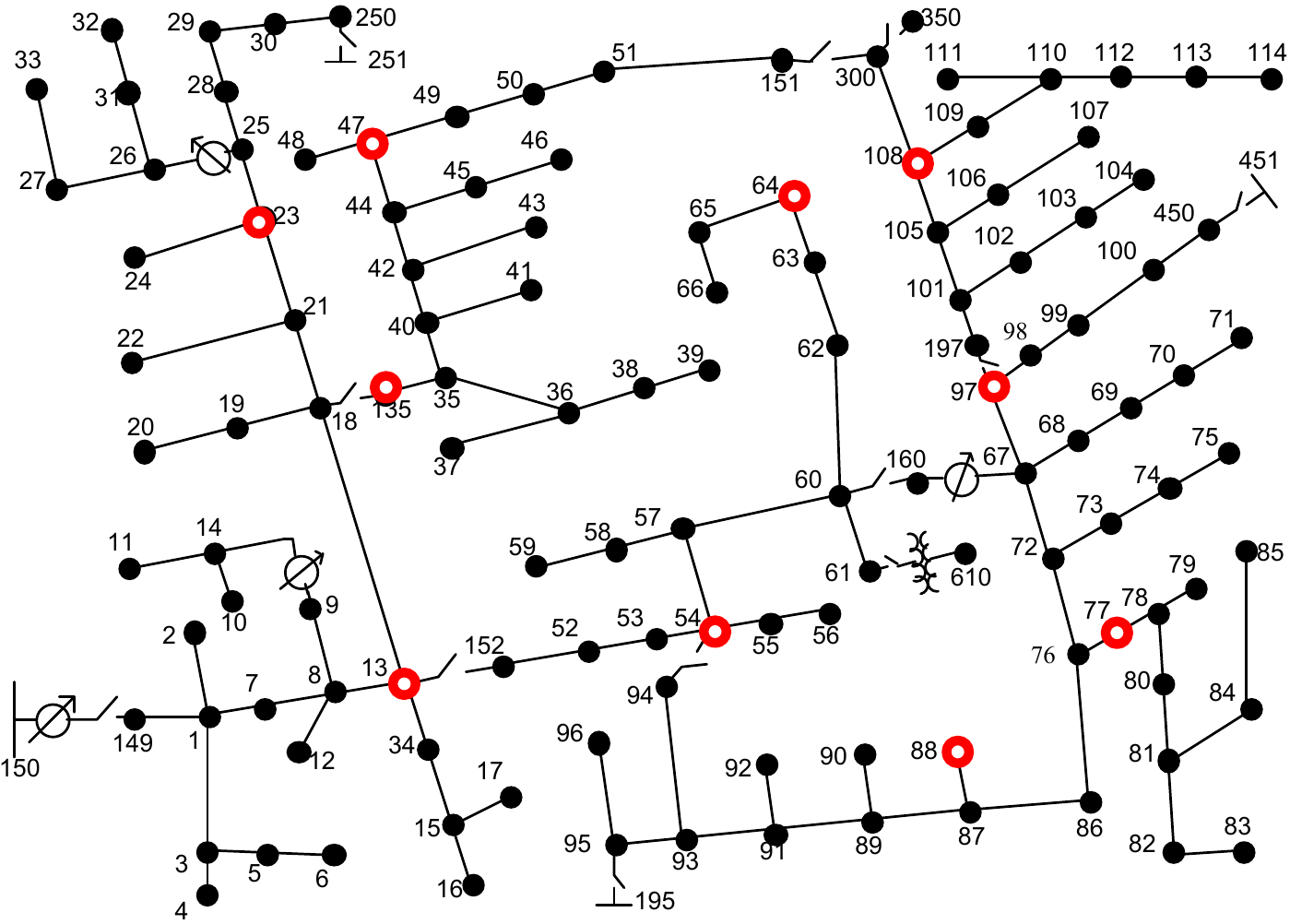}
	\caption{IEEE 123-bus system}
	\label{IEEE123}
\end{figure}

\begin{figure}[t]
	\setlength{\abovecaptionskip}{0pt}
	\centering
	\includegraphics[width=0.38\textwidth]{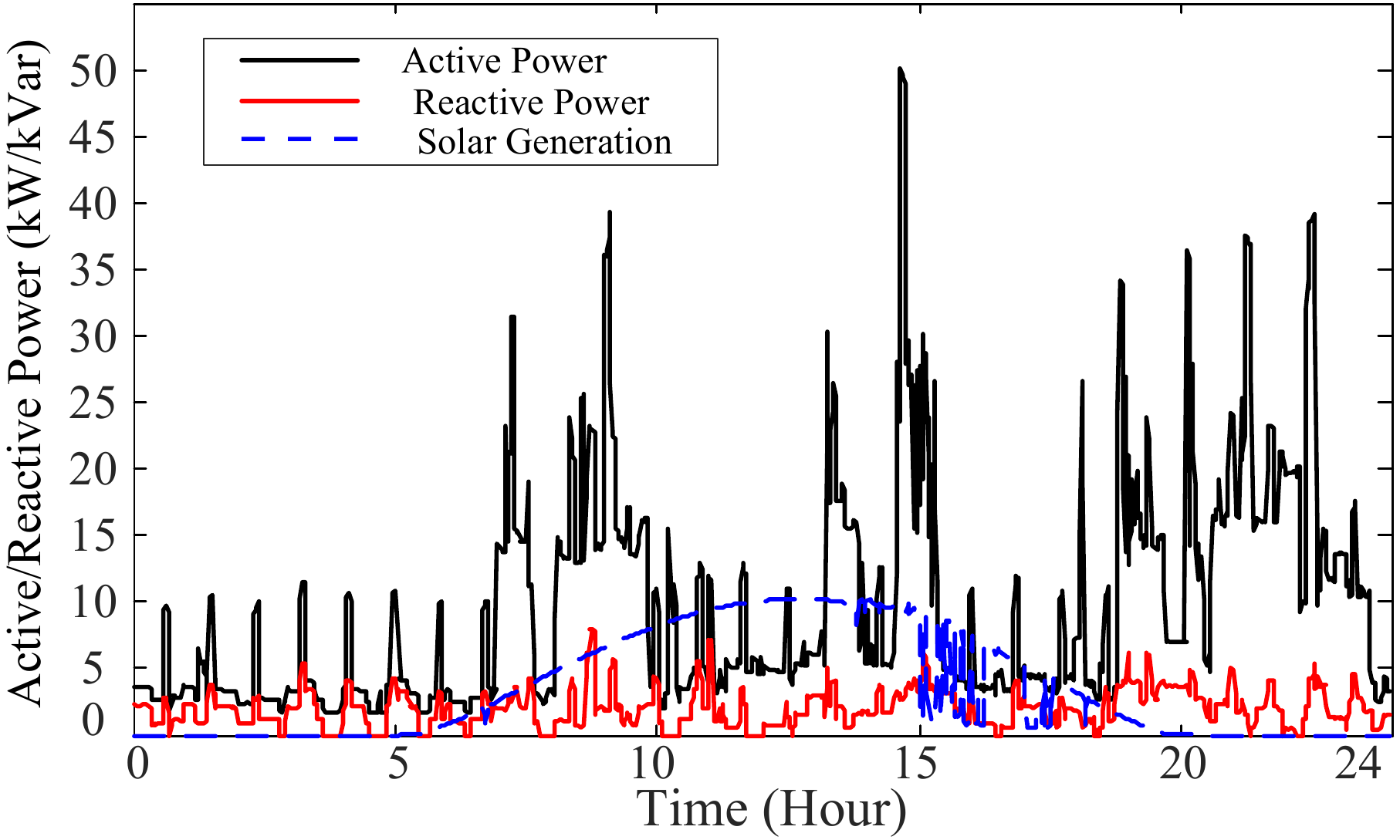}
	\caption{Active, reactive loads and solar generation within 24 hours}
	\label{load_demand}
\end{figure}

To validate the performance of the ASDVC in applications, we compare the results with SDVC and ASDVC under random time delays. The maximal time delay is $5$s. The profiles of daily network-wide voltage error with ASDVC and SDVC are given in Fig.\ref{Case123_results2}.
\begin{figure}[t]
	\setlength{\abovecaptionskip}{0pt}
	\centering
	\includegraphics[width=0.4\textwidth]{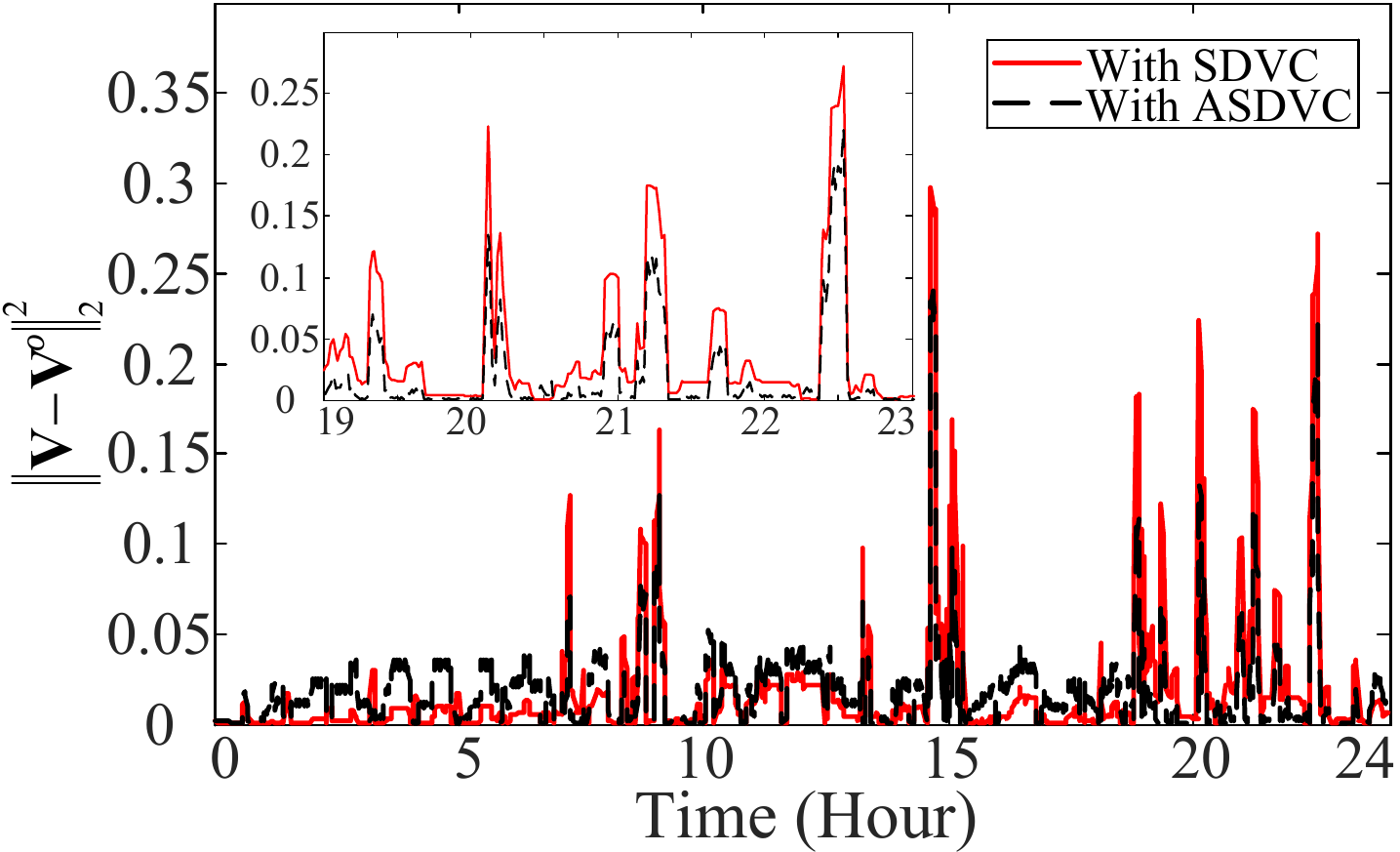}
	\caption{Daily voltage mismatch error with SDVC and ASDVC under random time delays. If SDVC is adopted, every bus has to wait for the slowest one to proceed the iteration. In contrast, every bus can update as long as new information is obtained in the ASDVC.}
	\label{Case123_results2}
\end{figure}
It is illustrated that the voltage deviation with SDVC is bigger than that with ASDVC if there exist time delays. The reason is that each bus under SDVC has to wait for the slowest neighbor to carry out the algorithm. In this situation, it cannot track system changes rapidly. It is different under ASDVC as there is no idling time for each bus. This shows that the ASDVC has better performance in time varying environments when time delays exist.

%Similarly, the voltage profiles of node 7 with SDVC and ASDVC are given in Fig.\ref{Case123_results4}, where the voltage drop with SDVC is bigger. This shows that the ASDVC has better convergence performance when time delays exist in the system.
%\begin{figure}[t]
%	\setlength{\abovecaptionskip}{0pt}
%	\centering
%	\includegraphics[width=0.34\textwidth]{Case123_results4}
%	\caption{Daily voltage profile with SDVC and ASDVC under random time delays}
%	\label{Case123_results4}
%\end{figure}

\section{Conclusion}
In this paper, we have developed an asynchronous distributed control method to regulate the voltage in distribution networks by making use of both active and reactive controllable power of DERs. The partial primal-dual gradient algorithm is utilized to design the controller with  proofs of  convergence and optimality of the equilibrium. Finally, numerical tests on an 8-bus system verify the similar convergence speed of SDVC and ASDVC. The daily simulations in the IEEE 123-bus system with real data show that the voltage deviation can be reduced using ASDVC. Simulations under random time delays show that the asynchronous algorithm has better performance in time-varying environments. 
In the theoretic analysis, the distribution network is assumed to be  three-phase symmetric and homogeneous. How to  eliminating these restrictions is among our ongoing works.

%\section*{References}

\bibliography{mybib}

\end{document}